\documentclass[11pt]{article}
\usepackage{amsfonts,amsmath,amsthm,amssymb}

\usepackage{comment, enumerate}
\usepackage[margin=1in]{geometry}
\usepackage{soul,color}
\usepackage{mathrsfs}
\usepackage{mathdots}
\usepackage[linesnumbered,boxed,ruled,vlined]{algorithm2e}

\usepackage{bbm}

\usepackage{url}

\usepackage{tabu}

\usepackage{xfrac}

\usepackage[normalem]{ulem}

\newcommand{\poly}{\mathop{\mathrm{poly}}}

\newcommand{\sr}{\mathop{\operatorname{S}}}
\newcommand{\xr}{\mathop{\operatorname{S_x}}}
\newcommand{\yr}{\mathop{\operatorname{S_y}}}
\newcommand{\zr}{\mathop{\operatorname{S_z}}}
\newcommand{\rr}{\mathop{\operatorname{R}}}

\newcommand{\F}{\mathbb{F}}
\newcommand{\R}{\mathbb{R}}

\newcommand{\N}{\mathbb{N}}

\theoremstyle{plain}\newtheorem{theorem}{Theorem}[section]
\newtheorem{lemma}[theorem]{Lemma}
\newtheorem{proposition}[theorem]{Proposition}
\newtheorem{corollary}[theorem]{Corollary}

\newtheorem{fact}[theorem]{Fact}
\newtheorem{remark}[theorem]{Remark}
\theoremstyle{definition}
\newtheorem{definition}[theorem]{Definition}

\newtheorem{question}[theorem]{Question}

\usepackage[utf8]{inputenc}

\title{Limits on the Universal Method for Matrix Multiplication}
\author{Josh Alman\footnote{MIT CSAIL and EECS, jalman@mit.edu. Supported by two NSF Career Awards.}}

\begin{document}

\maketitle

\begin{abstract}
In this work, we prove limitations on the known methods for designing matrix multiplication algorithms. Alman and Vassilevska Williams~\cite{aw} recently defined the \emph{Universal Method}, which substantially generalizes all the known approaches including Strassen's Laser Method~\cite{strassenlaser1} and Cohn and Umans' Group Theoretic Method~\cite{cohn2003group}. We prove concrete lower bounds on the algorithms one can design by applying the Universal Method to many different tensors. Our proofs use new tools for upper bounding the \emph{asymptotic slice rank} of a wide range of tensors. Our main result is that the Universal method applied to any Coppersmith-Winograd tensor $CW_q$ cannot yield a bound on $\omega$, the exponent of matrix multiplication, better than $2.16805$. By comparison, it was previously only known that the weaker `Galactic Method' applied to $CW_q$ could not achieve an exponent of $2$.

We also study the Laser Method (which is, in principle, a highly special case of the Universal Method) and prove that it is ``complete'' for matrix multiplication algorithms: when it applies to a tensor $T$, it achieves $\omega = 2$ if and only if it is possible for the Universal method applied to $T$ to achieve $\omega = 2$. Hence, the Laser Method, which was originally used as an algorithmic tool, can also be seen as a lower bounding tool. For example, in their landmark paper, Coppersmith and Winograd~\cite{coppersmith} achieved a bound of $\omega \leq 2.376$, by applying the Laser Method to $CW_q$. By our result, the fact that they did not achieve $\omega=2$ \emph{implies} a lower bound on the Universal Method applied to $CW_q$. Indeed, if it were possible for the Universal Method applied to $CW_q$ to achieve $\omega=2$, then  Coppersmith and Winograd's application of the Laser Method would have achieved $\omega=2$.
\end{abstract}

\thispagestyle{empty}
\newpage
\setcounter{page}{1}

\section{Introduction}

One of the biggest open questions in computer science asks how quickly one can multiply two matrices. Progress on this problems is measured by giving bounds on $\omega$, the \emph{exponent of matrix multiplication}, defined as the smallest real number such that two $n \times n$ matrices over a field can be multiplied using $n^{\omega + o(1)}$ field operations. Since Strassen's breakthrough algorithm~\cite{strassen} showing that $\omega \leq \log_2(7) \approx 2.81$, there has been a long line of work, resulting in the current best bound of $\omega \leq 2.3729$~\cite{v12,legall}, and it is popularly conjectured that $\omega=2$.

The key to Strassen's algorithm is an algebraic identity showing how $2 \times 2 \times 2$ matrix multiplication can be computed surprisingly efficiently (in particular, Strassen showed that the $2 \times 2 \times 2$ matrix multiplication tensor has rank at most $7$; see Section~\ref{sectwo} for precise definitions). Arguing about the ranks of larger matrix multiplication tensors has proven to be quite difficult -- in fact, even the rank of the $3 \times 3 \times 3$ matrix multiplication tensor isn't currently known. Progress on bounding $\omega$ since Strassen's algorithm has thus taken the following approach: Pick a tensor (trilinear form) $T$, typically not a matrix multiplication tensor, such that
\begin{itemize}
    \item Powers $T^{\otimes n}$ of $T$ can be efficiently computed (i.e. $T$ has low asymptotic rank), and
    \item $T$ is useful for performing matrix multiplication, since large matrix multiplication tensors can be `embedded' within powers of $T$.
\end{itemize}
Combined, these give an upper bound on the rank of matrix multiplication itself, and hence $\omega$.

The most general type of embedding which is known to preserve the ranks of tensors as required for the above approach is a \emph{degeneration}. In~\cite{aw}, the author and Vassilevska Williams called this method of taking a tensor $T$ and finding the best possible degeneration of powers $T^{\otimes n}$ into matrix multiplication tensors the \emph{Universal Method applied to $T$}, and the best bound on $\omega$ which can be proved in this way is written $\omega_u(T)$. They also defined two weaker methods: \emph{the Galactic Method applied to $T$}, in which the `embedding' must be a more restrictive \emph{monomial degeneration}, resulting in the bound $\omega_g(T)$ on $\omega$, and \emph{the Solar Method applied to $T$}, in which the `embedding' must be an even more restrictive \emph{zeroing out}, resulting in the bound $\omega_s(T)$ on $\omega$. Since monomial degenerations and zeroing outs are successively more restrictive types of degenerations, we have that for all tensors $T$, $$\omega \leq \omega_u(T) \leq \omega_g(T) \leq \omega_s(T).$$

These methods are \emph{very general}; there are no known methods for computing $\omega_u(T)$, $\omega_g(T)$, or $\omega_s(T)$ for a given tensor $T$, and these quantities are even unknown for very well-studied tensors $T$. The two main approaches to designing matrix multiplication algorithms are the Laser Method of Strassen~\cite{strassenlaser1} and the Group-Theoretic Method of Cohn and Umans~\cite{cohn2003group}. Both of these approaches show how to give upper bounds on $\omega_s(T)$ for particular structured tensors $T$ (and hence upper bound $\omega$ itself). In other words, they both give ways to find zeroing outs of tensors into matrix multiplication tensors, but not necessarily the best zeroing outs. In fact, it is known that the Laser Method does not always give the best zeroing out for a particular tensor $T$, since the improvements from~\cite{coppersmith} to later works~\cite{stothers,v12,legall} can be seen as giving slight improvements to the Laser Method to find better and better zeroing outs\footnote{These works apply the Laser Method to higher powers of the tensor $T = CW_q$, a technique which is still captured by the Solar Method.}. The Group-Theoretic Method, like the Solar Method, is very general, and it is not clear how to optimally apply it to a particular group or family of groups.

All of the improvements on bounding $\omega$ for the past 30+ years have come from studying the Coppersmith-Winograd family of tensors $\{CW_q\}_{q \in \N}$. The Laser Method applied to powers of $CW_5$ gives the bound $\omega_s(CW_5) \leq 2.3729$. The Group-Theoretic Method can also prove the best known bound $\omega \leq 2.3729$, by simulating the Laser Method analysis of $CW_q$ (see e.g. \cite{almanitcs} for more details). Despite a long line of work on matrix multiplication, there are no known tensors\footnote{The author and Vassilevska Williams~\cite{aw} study a generalization of $CW_q$ which can tie the best known bound, but its analysis is identical to that of $CW_q$. Our lower bounds in this paper will apply equally well to this generalized class as to $CW_q$ itself.} which seem to come close to achieving the bounds one can obtain using $CW_q$. This leads to the first main question of this paper: 
\begin{question}\label{q1}
How much can we improve our bound on $\omega$ using a more clever analysis of the Coppersmith-Winograd tensor?
\end{question}

The author and Vassilevska Williams~\cite{aw} addressed this question by showing that there is a constant $c>2$ so that for all $q$, $\omega_g(CW_q) > c$. In other words, the Galactic Method (monomial degenerations) cannot be used with $CW_q$ to prove $\omega=2$. However, this leaves open a number of important questions: How close to $2$ can we get using monomial degenerations; could it be that $\omega_g(CW_q) \leq 2.1$? Perhaps more importantly, what if we are allowed to use arbitrary degenerations; could it be that $\omega_u(CW_q) \leq 2.1$, or even $\omega_u(CW_q) = 2$?

The second main question of this paper concerns the Laser Method. The Laser Method upper bounds $\omega_s(T)$ for any tensor $T$ with certain structure (which we describe in detail in Section~\ref{lasersec}), and has led to every improvement on $\omega$ since its introduction by Strassen~\cite{strassenlaser1}.
\begin{question}\label{q2}
When the Laser Method applies to a tensor $T$, how close does it come to optimally analyzing $T$?
\end{question}
As discussed, we know the Laser Method does not always give a tight bound on $\omega_s(T)$. For instance, Coppersmith-Winograd~\cite{coppersmith} applied the Laser Method to $CW_q$ to prove $\omega_s(CW_q) \leq 2.376$, and then later work~\cite{stothers,v12,legall} analyzed higher and higher powers of $CW_q$ to show $\omega_s(CW_q) \leq 2.373$. Ambainis, Filmus and Le Gall~\cite{ambainis} showed that analyzing higher and higher powers of $CW_q$ itself with the Laser Method cannot yield an upper bound better than $\omega_s(CW_q) \leq 2.3725$. What about for other tensors? Could there be a tensor such that applying the Laser Method to $T$ yields $\omega_s(T) \leq c$ for some $c > 2$, but applying the Laser Method to high powers $T^{\otimes n}$ of $T$ yields $\omega_s(T) = 2$? Could applying an entirely different method to such a $T$, using arbitrary degenerations and not just zeroing outs, show that $\omega_u(T) = 2$?

\subsection{Our Results}

We give strong resolutions to both Question~\ref{q1} and Question~\ref{q2}.

\paragraph{Universal Method Lower Bounds}
To resolve Question~\ref{q1}, we prove a new lower bound for the Coppersmith-Winograd tensor:

\begin{theorem} \label{mainthm}
$\omega_u(CW_q) \geq 2.16805$ for all $q$.
\end{theorem}

In other words, no analysis of $CW_q$, using any techniques within the Universal Method, can prove a bound on $\omega$ better than $2.16805$. This generalizes the main result of~\cite{aw} from the Galactic method to the \emph{Universal} method, and gives a more concrete lower bound, increasing the bound from `a constant greater than $2$' to $2.16805$. We also give stronger lower bounds for particular tensors in the family. For instance, for the specific tensor $CW_5$ which yields the current best bound on $\omega$, we show $\omega_u(CW_5) \geq 2.21912\ldots$.

Our proof of Theorem~\ref{mainthm} proceeds by upper bounding $\tilde{S}(CW_q)$, the \emph{asymptotic slice rank} of $CW_q$. The slice rank of a tensor, denoted $S(T)$, was first introduced by Blasiak et al.~\cite{blasiak} in the context of lower bounds against the Group-Theoretic Method. In order to study degenerations of \emph{powers} of tensors, rather than just tensors themselves, we need to study an \emph{asymptotic} version of slice rank, $\tilde{S}$. This is important since the slice rank of a product of two tensors can be greater than the product of their slice ranks, and as we will show, $S(CW_q^{\otimes n})$ is much greater than $S(CW_q)^n$ for big enough $n$.

We will give three different tools for proving upper bounds on $\tilde{S}(T)$ for many different tensors $T$. These, combined with the known connection, that upper bounds on the slice rank of $T$ yield lower bounds on $\omega_u(T)$, will imply our lower bound for $CW_q$ as well as many other tensors of interest, including: the same lower bound $\omega_u(CW_{q,\sigma}) \geq 2.16805$ for any \emph{generalized Coppersmith-Winograd tensor} $CW_{q,\sigma}$ as introduced in~\cite{aw}, a similar lower bound for $cw_{q,\sigma}$, the generalized `simple' Coppersmith-Winograd tensor missing its `corner terms', and a lower bound for $T_q$, the structural tensor of the cyclic group $C_q$, matching the lower bounds obtained by \cite{almanitcs, blasiak}.
In Section~\ref{secfour} we give tables of our precise lower bounds for these and other tensors.

The Galactic Method lower bounds of~\cite{aw} were proved by introducing a suite of tools for giving upper bounds on $\tilde{I}(T)$, the \emph{asymptotic independence number} (sometimes also called the `galactic subrank' or the `monomial degeneration subrank') for many tensors $T$. We will show that our new tools are able to prove at least as high a lower bound on $\tilde{S}(T)$ as the tools of~\cite{aw} can prove on $\tilde{I}(T)$. We thus show that all of those previously known Galactic Method lower bounds hold for the Universal Method as well.

We also show how our slice rank lower bounds can be used to study other properties of tensors. Coppersmith and Winograd~\cite{coppersmith} introduced the notion of the \emph{value} $V_\tau(T)$ of a tensor $T$, which is useful when applying the Laser Method to a larger tensor $T'$ which contains $T$ as a subtensor. We show how our slice rank lower bounding tools yield a tight upper bound on the value of $t_{112}$, the notorious subtensor of $CW_q^{\otimes 2}$ which arises when applying the Laser Method to powers of $CW_q$. Although the value $V_\tau(t_{112})$ appears in every analysis of $CW_q$ since~\cite{coppersmith}, including~\cite{stothers,v12,legall,legallrect,legallrect2}, the best lower bound on it has not improved since~\cite{coppersmith}, and our new upper bound here helps explain why. See Sections~\ref{subsec:value}~and~\ref{t112} for more details.

We briefly note that our lower bound of $2.16805 > 2 + \frac16$ in Theorem~\ref{mainthm} may be significant when compared to the recent algorithm of Cohen, Lee and Song~\cite{cohen2018solving} which solves $n$-variable linear programs in time about $O(n^\omega + n^{2 + 1/6})$.

\paragraph{The Laser Method is ``Complete''}
We also show that for a wide class of tensors $T$, including $CW_q$, $cw_q$, $T_q$, and all the other tensors we study in Section~\ref{secfour}, our tools are tight, meaning they not only give an upper bound on $\tilde{S}(T)$, but they also give a matching lower bound. Hence, for these tensors $T$, no better lower bound on $\omega_u(T)$ is possible by arguing only about $\tilde{S}(T)$. 

The tensors we prove this for are what we call \emph{laser-ready} tensors -- tensors to which the Laser Method (as used by \cite{coppersmith} on $CW_q$) applies; see Definition~\ref{laserable} for the precise definition. Tensors need certain structure to be laser-ready, but tensors $T$ with this structure are essentially the only ones for which successful techniques for upper bounding $\omega_u(T)$ are known. In fact, every record-holding tensor in the history of matrix multiplication algorithm design has been laser-ready.

We show that for any laser-ready tensor $T$, the Laser Method can be used to construct a degeneration from $T^{\otimes n}$ to an independent tensor of size $\Lambda^{n - o(n)}$, where $\Lambda$ is the upper bound on $\tilde{S}(T)$ implied by one of our tools, Theorem~\ref{two}. Combined, these imply that $\tilde{S}(T) = \Lambda$, showing that the lower bound from Theorem~\ref{two} is tight. This gives an intriguing answer to Question~\ref{q2}:

\begin{theorem} \label{introlaseropt}
If $T$ is a laser-ready tensor, and the Laser Method applied to $T$ yields the bound $\omega_u(T) \leq c$ for some $c>2$, then $\omega_u(T)>2$.
\end{theorem}

To reiterate: If $T$ is any tensor to which the Laser Method applies (as in Definition~\ref{laserable}), and the Laser Method does not yield $\omega=2$ when applied to $T$, then in fact $\omega_u(T)>2$, and even the substantially more general Universal method applied to $T$ cannot yield $\omega=2$. Hence, the Laser Method, which was originally used as an algorithmic tool, can also be seen as a lower bounding tool. Conversely, Theorem~\ref{introlaseropt} shows that the Laser Method is ``complete'', in the sense that it cannot yield a bound on $\omega$ worse than $2$ when applied to a tensor which is able to prove $\omega=2$.

Theorem~\ref{introlaseropt} explains and generalizes a number of phenomena:
\begin{itemize}
    \item The fact that Coppersmith-Winograd~\cite{coppersmith} applied the Laser method to the tensor $CW_q$ and achieved an upper bound greater than $2$ on $\omega$ \emph{implies that} $\omega_u(CW_q) > 2$, and no \emph{arbitrary degeneration} of powers of $CW_q$ can yield $\omega = 2$.
    
    \item As mentioned above, it is known that applying the Laser method to higher and higher powers of a tensor $T$ can successively improve the resulting upper bound on $\omega$. Theorem~\ref{introlaseropt} shows that if the Laser method applied to the first power of any tensor $T$ did not yield $\omega=2$, then this sequence of Laser method applications (which is a special case of the Universal method) must converge to a value greater than $2$ as well. This generalizes the result of Ambainis, Filmus and Le Gall~\cite{ambainis}, who proved this about applying the Laser Method to higher and higher powers of the specific tensor $T = CW_q$.
    
    \item Our result also generalizes the result of Kleinberg, Speyer and Sawin~\cite{kleinberg}, where it was shown that (what can be seen as) the Laser method achieves a tight lower bound on $\tilde{S}(T_q^{lower})$, matching the upper bound of Blasiak et al.~\cite{blasiak}. Indeed, $T_q^{lower}$, the lower triangular part of $T_q$, is a laser-ready tensor.
\end{itemize}

Our proof of Theorem~\ref{introlaseropt} also sheds light on a notion related to the asymptotic slice rank $\tilde{S}(T)$ of a tensor $T$, called the \emph{asymptotic subrank} $\tilde{Q}(T)$ of $T$. $\tilde{Q}$ is a ``dual'' notion of asymptotic rank, and it is important in the definition of Strassen's asymptotic spectrum of tensors~\cite{strassenlaser1}. 

It is not hard to see (and follows, for instance, from Propositions~\ref{propone}~and~\ref{proptwo} below) that $\tilde{Q}(T) \leq \tilde{S}(T)$ for all tensors $T$. However, there are no known separations between the two notions; whether there exists a tensor $T$ such that $\tilde{Q}(T) < \tilde{S}(T)$ is an open question. As a Corollary of Theorem~\ref{introlaseropt}, we prove:
\begin{corollary} \label{introsubrankcor}
Every laser-ready tensor $T$ has $\tilde{Q}(T) = \tilde{S}(T)$.
\end{corollary}
\noindent Since, as discussed above, almost all of the most-studied tensors are laser-ready, this might help explain why we have been unable to separate the two notions.

\subsection{Other Related Work}

\paragraph{Probabilistic Tensors and Support Rank}
Cohn and Umans~\cite{cohn2013fast} introduced the notion of the \emph{support rank} of tensors, and showed that upper bounds on the support rank of matrix multiplication tensors can be used to design faster \emph{Boolean} matrix multiplication algorithms. Recently, Karppa and Kaski~\cite{karppa2019probabilistic} used `probabilistic tensors' as another way to design Boolean matrix multiplication algorithms. 

In fact, our tools for proving asymptotic slice rank upper bounds can be used to prove lower bounds on these approaches as well. For instance, our results imply that finding a `weighted' matrix multiplication tensor as a degeneration of a power of $CW_q$ (in order to prove a support rank upper bound) cannot result in a better exponent for Boolean matrix multiplication than $2.16805$. 

This is because `weighted' matrix multiplication tensors can degenerate into independent tensors just as large as their unweighted counterparts. Similarly, if a probabilistic tensor $\mathcal{T}$ is degenerated into a (probabilistic) matrix multiplication tensor, Karppa and Kaski show that this gives a corresponding support rank expression for matrix multiplication as well, and so upper bounds on $\tilde{S}(T)$ for any $T$ in the support of $\mathcal{T}$ also result in lower bounds on this approach. 

\paragraph{Concurrent Work} Christandl, Vrana and Zuiddam~\cite{theirpreprint} independently proved some of the same lower bounds on $\omega_u$ as us, including Theorem~\ref{mainthm}.
Although we achieve the same upper bounds on $\omega_u(T)$ for a number of tensors, our techniques seem different: we use simple combinatorial tools generalizing those from our prior work~\cite{aw}, 
while their bounds use the seemingly more complicated machinery of Strassen's asymptotic spectrum of tensors~\cite{strassen1991degeneration}. They thus phrase their results in terms of the asymptotic subrank $\tilde{Q}(T)$ of tensors rather than the asymptotic slice rank $\tilde{S}(T)$, and the fact that their bounds are often the same as ours is related to the fact we prove, in Corollary~\ref{introsubrankcor}, that $\tilde{Q}(T) = \tilde{S}(T)$ for all of the tensors we study; see the bottom of Section~\ref{sliceranksubsec} for a more technical discussion of the differences between the two notions. Our other results and applications of our techniques are, as far as we know, entirely new, including our matching lower bounds for $\tilde{S}(CW_q)$, $\tilde{S}(cw_q)$, and $\tilde{S}(T_q)$, bounding the value $V_\tau(T)$ of tensors, and all our results about the completeness of the Laser Method. By comparison, their `irreversibility' approach only seems to upper bound $\omega_u(T)$ itself.

\subsection{Outline}
In Section~\ref{overviewsec} we give an overview of the proofs of our main results. In Section~\ref{sectwo} we introduce all the concepts and notation related to tensors which will be used throughout the paper. In particular, in Subsection~\ref{sliceranksubsec} we introduce the relevant notions and basic properties related to slice rank. In Section~\ref{secthree} we present the proofs of our new lower bounding tools for asymptotic slice rank. In Section~\ref{secfour} we apply these tools to a number of tensors of interest including $CW_q$. Finally, in Section~\ref{lasersec}, we define and discuss the ``completeness'' of the Laser method.

\section{Proof Overview} \label{overviewsec}

We give a brief overview of the techniques we use to prove our main results, Theorems~\ref{mainthm}~and~\ref{introlaseropt}. All the technical terms we refer to here will be precisely defined in Section~\ref{sectwo}.

\paragraph{Section~\ref{sliceranksubsec}: Asymptotic Slice Rank and its Connection with Matrix Multiplication} The tensors we study are 3-tensors, which can be seen as trilinear forms over three sets $X,Y,Z$ of formal variables. The slice rank $S(T)$ of a tensor $T$ is a measure of the complexity of $T$, analogous to the rank of a matrix. In this paper we study the \emph{asymptotic slice rank} $\tilde{S}(T)$ of tensors $T$:
$$\tilde{S}(T) := \limsup_{n \in\N} S(T^{\otimes n})^{1/n}.$$
$\tilde{S}$ satisfies two key properties:
\begin{enumerate}
    \item Degenerations cannot increase the asymptotic slice rank of a tensor. In other words, if $A$ degenerates to $B$, then $\tilde{S}(B) \leq \tilde{S}(A)$.
    \item Matrix multiplication tensors have high asymptotic slice rank.
\end{enumerate}
This means that if a certain tensor $T$ has a small value of $\omega_u(T)$, or in other words, powers $T^{\otimes n}$ can degenerate into large matrix multiplication tensors, then $T$ itself must have large asymptotic slice rank. Hence, in order to lower bound $\omega_u(T)$, it suffices to upper bound $\tilde{S}(T)$.

\paragraph{Section~\ref{secthree}: Tools for Upper Bounding Asymptotic Slice Rank} In general, bounding $\tilde{S}(T)$ for a tensor $T$ can be much more difficult than bounding $S(T)$. This is because $S$ can be supermultiplicative, i.e. there are tensors $A$ and $B$ such that $S(A) \cdot S(B) \ll S(A \otimes B)$. Indeed, we will show that $\tilde{S}(T) > S(T)$ for many tensors $T$ of interest, including $T = CW_q$.

We will give three new tools for upper bounding $\tilde{S}(T)$ for many tensors $T$. Each applies to tensors with different properties:

\begin{itemize}
    \item Theorem~\ref{one}: If $T$ is over $X,Y,Z$, then it is straightforward to see that if one of the variable sets is not too large, then $\tilde{S}(T)$ must be small: $\tilde{S}(T) \leq \min\{ |X|, |Y|, |Z| \}$. In this first tool, we show how if $T$ can be written as a sum $T = T_1 + \cdots + T_k$ of a few tensors, and each $T_i$ does not have many of one type of variable, then we can still derive an upper bound on $\tilde{S}(T)$. 
    \item Theorem~\ref{two}: The second tool concerns partitions of the variable sets $X,Y,Z$. It shows that if $\tilde{S}(T)$ is large, then there is a probability distribution on the blocks of $T$ (subtensors corresponding to a choice of one part from each of the three partitions) so that the total probability mass assigned to each part of each partition is proportional to its size. Loosely, this means that $T$ must have many different `symmetries', no matter how its variables are partitioned. 
    \item Theorem~\ref{removeanxtool}: Typically, for tensors $A$ and $B$, even if $\tilde{S}(A)$ and $\tilde{S}(B)$ are `small', it may still be the case that $\tilde{S}(A+B)$ is large. This third tool shows that if $A$ has an additional property, then one can still bound $\tilde{S}(A+B)$. Roughly, the property that $A$ must satisfy is that not only is $\tilde{S}(A)$ small, but a related notion called the `x-rank' of $A$ must also be small.
\end{itemize}

In particular, we will remark that our three tools for bounding $\tilde{S}(T)$ strictly generalize similar tools introduced by~\cite{aw} for bounding $\tilde{I}(T)$. Hence, we generalize their results bounding $\omega_g(T)$ for various tensors $T$ to bounds on $\omega_u(T)$.

\paragraph{Section~\ref{secfour}: Universal Method Lower Bounds} We apply our tools to prove upper bounds on $\tilde{S}(T)$, and hence lower bounds on $\omega_u(T)$, for a number of tensors $T$ of interest. To prove Theorem~\ref{mainthm}, we show that \emph{all three} tools can be applied to $CW_q$. We also apply our tools to many other tensors of interest including the generalized Coppersmith-Winograd tensors $CW_{q,\sigma}$, the generalized small Coppersmith-Winograd tensors $cw_{q,\sigma}$, the structural tensor $T_q$ of the cyclic group $C_q$ as well as its `lower triangular version' $T_q^{lower}$ , and the subtensor $t_{112}$ of $CW_q^{\otimes 2}$ which arises in~\cite{coppersmith,stothers,v12,legall,legallrect,legallrect2}. Throughout Section~\ref{secfour} we give many tables of concrete lower bounds that we prove for the tensors in all these different families.

\paragraph{Section~\ref{lasersec}: ``Completeness'' of the Laser Method} Finally, we study the Laser Method. The Laser Method applied to a tensor $T$ shows that powers $T^{\otimes n}$ can zero out into large matrix multiplication tensors. Using the properties of $\tilde{S}$ that we prove in Section~\ref{sliceranksubsec}, we will show that the Laser Method can also be applied to a tensor $T$ to prove a lower bound on $\tilde{S}(T)$. (More precisely, it actually proves a lower bound on $\tilde{Q}(T)$, the \emph{asymptotic subrank of $T$}, which in turn lower bounds $\tilde{S}(T)$).

We prove Theorem~\ref{introlaseropt} by combining this construction with Theorem~\ref{two}, one of our tools for upper bounding $\tilde{S}(T)$. Intuitively, both Theorem~\ref{two} and the Laser Method are concerned with probability distributions on blocks of a tensor, and both involve counting the number of variables in powers $T^{\otimes n}$ that are consistent with these distributions. We use this intuition to show that the upper bound given by Theorem~\ref{two} is equal to the lower bound given by the Laser Method.

\section{Preliminaries} \label{sectwo}

We begin by introducing the relevant notions and notation related to tensors and matrix multiplication. We will use the same notation introduced in \cite[Section~3]{aw}, and readers familiar with that paper may skip to Subsection~\ref{subsec:value}.

\subsection{Tensor Basics}

For sets $X = \{x_1, \ldots, x_q\}$, $Y = \{y_1, \ldots, y_r\}$, and $Z = \{z_1, \ldots, z_s \}$ of formal variables, a \emph{tensor over $X,Y,Z$} is a trilinear form
$$T = \sum_{x_i \in X, y_j \in Y, z_k \in Z} \alpha_{ijk} x_i y_j z_k,$$
where the $\alpha_{ijk}$ coefficients come from an underlying field $\F$. The \emph{terms}, which we write as $x_i y_j z_k$, are sometimes written as $x_i\otimes y_j\otimes z_k$ in the literature. We say $T$ is \emph{minimal for} $X,Y,Z$ if, for each $x_i \in X$, there is a term involving $x_i$ with a nonzero coefficient in $T$, and similarly for $Y$ and $Z$ (i.e. $T$ can't be seen as a tensor over a strict subset of the variables). We say that two tensors $T_1, T_2$ are \emph{isomorphic}, written $T_1 \simeq T_2$, if they are equal up to renaming variables.

If $T_1$ is a tensor over $X_1, Y_1, Z_1$, and $T_2$ is a tensor over $X_2, Y_2, Z_2$, then the \emph{tensor product} $T_1 \otimes T_2$ is a tensor over $X_1 \times X_2, Y_1 \times Y_2, Z_1 \times Z_2$ such that, for any $(x_1,x_2) \in X_1 \times X_2$, $(y_1,y_2) \in Y_1 \times Y_2$, and $(z_1,z_2) \in  Z_1 \times Z_2$, the coefficient of $(x_1,x_2)(y_1,y_2)(z_1,z_2)$ in $T_1 \otimes T_2$ is the product of the coefficient of $x_1 y_1 z_1$ in $T_1$, and the coefficient of $x_2 y_2 z_2$ in $T_2$. 
For any tensor $T$ and positive integer $n$, the tensor power $T^{\otimes n}$ is the tensor over $X^n, Y^n, Z^n$ resulting from taking the tensor product of $n$ copies of $T$.

If $T_1$ is a tensor over $X_1, Y_1, Z_1$, and $T_2$ is a tensor over $X_2, Y_2, Z_2$, then the \emph{direct sum} $T_1 \oplus T_2$ is a tensor over $X_1 \sqcup X_2$, $Y_1 \sqcup Y_2$, $Z_1 \sqcup Z_2$ which results from forcing the variable sets to be disjoint (as in a normal disjoint union) and then summing the two tensors. For a nonnegative integer $m$ and tensor $T$ we write $m \odot T$ for the disjoint sum of $m$ copies of $T$.

\subsection{Tensor Rank}

A tensor $T$ has \emph{rank one} if there are values $a_i \in \F$ for each $x_i \in X$, $b_j \in \F$ for each $y_j \in Y$, and $c_k \in \F$ for each $z_k \in Z$, such that the coefficient of $x_i y_j z_k$ in $T$ is $a_i b_j c_k$, or in other words,
$$T=\sum_{x_i \in X, y_j \in Y, z_k \in Z} a_i b_j c_k \cdot x_i y_j z_k = \left( \sum_{x_i \in X} a_i x_i \right) \left( \sum_{y_j \in Y} b_j y_j \right) \left( \sum_{z_k \in Z} c_k z_k \right).$$
The \emph{rank} of a tensor $T$, denoted $R(T)$, is the smallest number of rank one tensors whose sum (summing the coefficient of each term individually) is $T$. It is not hard to see that for tensors $T$ and positive integers $n$, we always have $R(T^{\otimes n}) \leq R(T)^n$, but for some tensors $T$ of interest this inequality is not tight. We thus define the \emph{asymptotic rank} of tensor $T$ as $\tilde{R}(T) := \liminf_{n \in \N} (R(T^{\otimes n}))^{1/n}$.

\subsection{Matrix Multiplication Tensors}

For positive integers $a,b,c$, the matrix multiplication tensor $\langle a,b,c \rangle$ is a tensor over $\{x_{ij}\}_{i \in [a], j \in [b]}$, $\{y_{jk}\}_{j \in [b], k \in [c]}$, $\{ z_{ki} \}_{k \in [c], i \in [a]}$ given by
$$\langle a,b,c\rangle = \sum_{i=1}^a\sum_{j=1}^b\sum_{k=1}^c x_{ij}y_{jk}z_{ki}.$$
It is not hard to verify that for positive integers $a_1, a_2, b_1, b_2, c_1, c_2$, we have $\langle a_1, b_1, c_1 \rangle \otimes \langle a_2, b_2, c_2 \rangle \simeq \langle a_1 a_2, b_1 b_2, c_1 c_2 \rangle$. The \emph{exponent of matrix multiplication}, denoted $\omega$, is defined as $$\omega := \liminf_{a,b,c \in \N} 3 \log_{abc}(R(\langle a,b,c \rangle)).$$
Because of the tensor product property above, we can alternatively define $\omega$ in a number of ways:
$$\omega = \liminf_{a,b,c \in \N} 3 \log_{abc}(\tilde{R}(\langle a,b,c \rangle)) = \liminf_{n \in \N} \log_n \tilde{R}(\langle n,n,n \rangle) = \log_2(\tilde{R}(\langle 2,2,2 \rangle)).$$
For instance, Strassen~\cite{strassen} showed that $R(\langle 2,2,2 \rangle) \leq 7$, which implies that $\omega \leq \log_2(7)$.
\subsection{Degenerations and the Universal Method}

We now describe a very general way to transform from a tensor $T_1$ over $X_1, Y_1, Z_1$ to a tensor $T_2$ over $X_2, Y_2, Z_2$. For a formal variable $\lambda$, pick maps $\alpha : X_1 \times X_2 \to \F(\lambda)$, $\beta : Y_1 \times Y_2 \to \F(\lambda)$, and $\gamma : Z_1 \times Z_2 \to \F(\lambda)$, which map pairs of variables to polynomials in $\lambda$, and pick an integer $h$. Then, when you replace each $x \in X_1$ with $\sum_{x' \in X_2} \alpha(x,x') x'$, each $y \in Y_1$ with $\sum_{y' \in Y_2} \beta(y,y') y'$, and each $z \in Z_1$ with $\sum_{z' \in Z_2} \gamma(z,z') z'$, in $T_1$, then the resulting tensor $T'$ is a tensor over $X_2, Y_2, Z_2$ with coefficients over $\F(\lambda)$. When $T'$ is instead viewed as a polynomial in $\lambda$ whose coefficients are tensors over $X_2,Y_2,Z_2$ with coefficients in $\F$, it must be that $T_2$ is the coefficient of $\lambda^h$, and the coefficient of $\lambda^{h'}$ is $0$ for all $h'<h$.

If such a transformation is possible, we say $T_2$ is a \emph{degeneration} of $T_1$. There are also two more restrictive types of degenerations:
\begin{itemize}
    \item $T_2$ is a \emph{monomial degeneration} of $T_1$ if such a transformation is possible where the polynomials in the ranges of $\alpha, \beta, \gamma$ have at most one monomial, and furthermore, for each $x \in X_1$ there is at most one $x' \in X_2$ such that $\alpha(x,x') \neq 0$, and similarly for $\beta$ and $\gamma$.\footnote{Some definitions of monomial degenerations do not have this second condition, or equivalently, consider a monomial degeneration to be a `restriction' composed with what we defined here. The distinction is not important for this paper, but we give this definition since it captures Strassen's monomial degeneration from matrix multiplication tensors to independent tensors~\cite{laser} (see also Proposition~\ref{propthree} below), and it is the notion that the prior work~\cite{aw} proved lower bounds against.}
    \item $T_2$ is a \emph{zeroing out} of $T_1$ if, in addition to the restrictions of a monomial degeneration, the ranges of $\alpha, \beta, \gamma$ must be $\{0,1\}$.
\end{itemize}

Degenerations are useful in the context of matrix multiplication algorithms because degenerations cannot increase the rank of a tensor. In other words, if $T_2$ is a degeneration of $T_1$, then $R(T_2) \leq R(T_1)$~\cite{bini1980border}. It is often hard to bound the rank of matrix multiplication tensors directly, so all known approaches proceed by bounding the rank of a different tensor $T$ and then showing that powers of $T$ degenerate into matrix multiplication tensors.

More precisely, all known approaches fall within the following method, which we call the \emph{Universal Method}~\cite{aw} applied to a tensor $T$ of asymptotic rank $R = \tilde{R}(T)$: Consider all positive integers $n$, and all ways to degenerate $T^{\otimes n}$ into a disjoint sum $\bigoplus_{i=1}^m \langle a_i, b_i, c_i \rangle$ of matrix multiplication tensors, resulting in an upper bound on $\omega$ by the asymptotic sum inequality~\cite{Sch81} of $\sum_{i=1}^m (a_i b_i c_i)^{\omega/3} \leq R^n$. Then, $\omega_u(T)$, the bound on $\omega$ from the Universal Method applied to $T$, is the $\liminf$ over all such $n$ and degenerations, of the resulting upper bound on $\omega$.

In~\cite{aw}, two weaker versions of the Universal Method are also defined: \emph{the Galactic Method}, in which the degeneration must be a monomial degeneration, resulting in a bound $\omega_g(T)$, and \emph{the Solar Method}, in which the degeneration must be a zeroing out, resulting in a bound $\omega_s(T)$. To be clear, all three of these methods are very general, and we don't know the values of $\omega_s(T)$, $\omega_g(T)$, or $\omega_u(T)$ for almost any nontrivial tensors $T$. In fact, all the known approaches to bounding $\omega$ proceed by giving upper bounds on $\omega_s(T)$ for some carefully chosen tensors $T$; the most successful has been the Coppersmith-Winograd family of tensors $T = CW_q$, which has yielded all the best known bounds on $\omega$ since the 80's~\cite{CoppersmithW82,stothers,v12,legall}. Indeed, the two most successful approaches, the Laser Method~\cite{strassenlaser1} and the Group-Theoretic Approach~\cite{cohn2003group} ultimately use zeroing outs of tensors. We refer the reader to~\cite[Sections 3.3 and 3.4]{aw} for more details on these approaches and how they relate to the notions used here.

\subsection{Tensor Value} \label{subsec:value}

Coppersmith and Winograd~\cite{coppersmith} defined the \emph{value} of a tensor in their analysis of the $CW_q$ tensor. For a tensor $T$, and any $\tau \in [2/3, 1]$, the $\tau$-value of $T$, denoted $V_\tau(T)$, is defined as follows: Consider all positive integers $n$, and all ways $\sigma$ to degenerate $T^{\otimes n}$ into a direct sum $\bigoplus_{i=1}^{q(\sigma)} \langle a_i^\sigma, b_i^\sigma, c_i^\sigma \rangle$ of matrix multiplication tensors. Then, $V_\tau(T)$ is given by $$V_\tau(T) := \limsup_{n, \sigma} \left( \sum_{i=1}^{q(\sigma)} (a_i^\sigma b_i^\sigma c_i^\sigma)^\tau \right)^{1/n}.$$
We can then equivalently define $\omega_u(T)$ as the $\liminf$ of $\omega_u$, over all $\omega_u \in [2,3]$ such that $V_{\omega_u / 3}(T) \geq \tilde{R}(T)$. We can see from the power mean inequality that $V_\tau(T) \geq V_{2/3}(T)^{3 \tau / 2}$ for all $\tau \in [2/3,1]$, although this bound is often not tight as there can be better degenerations of $T^{\otimes n}$ depending on the value of $\tau$.

\subsection{Asymptotic Slice Rank} \label{sliceranksubsec}

The main new notions we will need in this paper relate to the slice rank of tensors. We say a tensor $T$ over $X,Y,Z$ has \emph{x-rank} $1$ if it is of the form
$$T = \left( \sum_{x \in X} \alpha_x \cdot x \right) \otimes \left( \sum_{y \in Y} \sum_{z \in Z} \beta_{y,z} \cdot y \otimes z \right) = \sum_{x \in X, y \in Y, z \in Z} \alpha_x \beta_{y,z} \cdot xyz$$
for some choices of the $\alpha$ and $\beta$ coefficients over the base field. More generally, the x-rank of $T$, denoted $\xr(T)$, is the minimum number of tensors of x-rank 1 whose sum is $T$. We can similarly define the y-rank, $\yr$, and the z-rank, $\zr$. Then, the \emph{slice rank} of $T$, denoted $\sr(T)$, is the minimum $k$ such that there are tensors $T_X$, $T_Y$ and $T_Z$ with $T = T_X + T_Y + T_Z$ and $\xr(T_X) + \yr(T_Y) + \zr(T_Z) = k$.

Unlike tensor rank, the slice-rank is not submultiplicative in general, i.e. there are tensors $A$ and $B$ such that $\sr(A \otimes B) > \sr(A) \cdot \sr(B)$. For instance, it is not hard to see that $\sr(CW_5) = 3$, but since it is known~\cite{v12,legall} that $\omega_s(CW_5) \leq 2.373$, it follows (e.g. from Theorem~\ref{thm:symmSOmega} below) that $\sr(CW_q^{\otimes n}) \geq 7^{n \cdot 2/2.373 - o(n)} \geq 5.15^{n - o(n)}$. We are thus interested in the \emph{asymptotic slice rank}, $\tilde{S}(T)$, of tensors $T$, defined as
$$\tilde{S}(T) := \limsup_{n \in \N} [\sr(T^{\otimes n})]^{1/n}.$$

We note a few simple properties of slice rank which will be helpful in our proofs:

\begin{lemma} \label{propertieslemma}
For tensors $A$ and $B$:

\begin{enumerate}[(1)]
    \item $\sr(A) \leq \xr(A) \leq \rr(A)$,
    \item $\xr(A \otimes B) \leq \xr(A) \cdot \xr(B)$,
    \item $\sr(A + B) \leq \sr(A) + \sr(B)$, and $\xr(A+B) \leq \xr(A) + \xr(B)$,
    \item $\sr(A \otimes B) \leq \sr(A) \cdot \max \{ \xr(B), \yr(B), \zr(B) \}$, and
    \item If $A$ is a tensor over $X,Y,Z$, then $\xr(T) \leq |X|$ and hence $\sr(T) \leq \min\{ |X|, |Y|, |Z|\}$.
\end{enumerate}
\end{lemma}

\begin{proof}
(1) and (2) are straightforward. (3) follows since the sum of the slice rank (resp. x-rank) expressions for $A$ and for $B$ gives a slice rank (resp. x-rank) expression for $A+B$. To prove (4), let $m = \max \{ \xr(B), \yr(B), \zr(B) \}$, and note that if $A = A_X + A_Y + A_Z$ such that $\xr(A_X) + \yr(A_Y) + \zr(A_Z) = \sr(A)$, then
$$A \otimes B = A_X\otimes B + A_Y\otimes B + A_Z\otimes B,$$
and so
\begin{align*}\sr(A \otimes B) &\leq \sr(A_X\otimes B) + \sr(A_Y\otimes B) + \sr(A_Z\otimes B) \\ &\leq \xr(A_X\otimes B) + \yr(A_Y\otimes B) + \zr(A_Z\otimes B)\\ &\leq \xr(A_X) \xr( B) + \yr(A_Y) \yr( B) + \zr(A_Z) \zr( B)\\ &\leq \xr(A_X) m + \yr(A_Y) m + \zr(A_Z) m = \sr(A) \cdot m.\end{align*}

Finally, (5) follows since, for instance, any tensor with one only x-variable has x-rank 1.
\end{proof}

Asymptotic slice rank is interesting in the context of matrix multiplication algorithms because of the following facts.

\begin{definition}
For a positive integer $q$, the \emph{independent tensor of size $q$}, denoted $\langle q \rangle$, is the tensor $\sum_{i=1}^q x_i y_i z_i$ with $q$ terms that do not share any variables.
\end{definition}

\begin{proposition}[\cite{tao2} Corollary 2] \label{propone}
If $A$ and $B$ are tensors such that $A$ has a degeneration to $B$, then $\sr(B) \leq \sr(A)$, and hence $\tilde{S}(B) \leq \tilde{S}(A)$.
\end{proposition}

\begin{proposition}[\cite{tao1} Lemma 1; see also \cite{blasiak} Lemma 4.7] \label{proptwo}
For any positive integer $q$, we have $\sr(\langle q \rangle) = \tilde{S}(\langle q \rangle) = q$, where $\langle q \rangle$ is the independent tensor of size $q$.
\end{proposition}

\begin{proposition}[\cite{laser} Theorem 4; see also \cite{aw} Lemma 4.2] \label{propthree}
For any positive integers $a,b,c$, the matrix multiplication tensor $\langle a,b,c \rangle$ has a (monomial) degeneration to an independent tensor of size at least $ 0.75 \cdot abc / \max \{a,b,c\}$.
\end{proposition}

\begin{corollary} \label{cortopropthree}
For any positive integers $a,b,c$, we have $\tilde{S}(\langle a,b,c \rangle) = abc / \max \{a,b,c\}$.
\end{corollary}
\begin{proof}
Assume without loss of generality that $c \geq a,b$. For any positive integer $n$, we have that $\langle a,b,c \rangle^{\otimes n} \simeq \langle a^n, b^n, c^n \rangle$ has a degeneration to an independent tensor of size at least $0.75 \cdot a^n b^n$, meaning $S(\langle a,b,c \rangle^{\otimes n}) \geq 0.75 \cdot a^n b^n$ and hence $\tilde{S}(\langle a,b,c \rangle) \geq (0.75)^{1/n} ab$, which means $\tilde{S}(\langle a,b,c \rangle) \geq  ab$. Meanwhile, $\langle a,b,c \rangle$ has $ab$ different $x$-variables, so it must have $\xr(\langle a,b,c \rangle) \leq ab$ and more generally, $\sr(\langle a,b,c \rangle^{\otimes n}) \leq \xr(\langle a,b,c \rangle^{\otimes n}) \leq (ab)^n$, which means $\tilde{S}(\langle a,b,c \rangle) \leq  ab$.
\end{proof}

To summarize: we know that degenerations cannot increase asymptotic slice rank, and that matrix multiplication tensors have a high asymptotic slice rank. Hence, if $T$ is a tensor such that $\omega_u(T)$ is `small', meaning a power of $T$ has a degeneration to a disjoint sum of many large matrix multiplication tensors, then $T$ itself must have `large' asymptotic slice rank. This can be formalized identically to~\cite[Theorem~4.1 and Corollary~4.3]{aw} to show:

\begin{theorem} \label{thm:SROmega}
For any tensor $T$, $$\tilde{S}(T) \geq \tilde{R}(T)^{\frac{6}{\omega_u(T)} - 2}.$$
\end{theorem}

\begin{corollary} \label{cor:omegaands}
For any tensor $T$, if $\omega_u(T)=2$, then $\tilde{S}(T) = \tilde{R}(T)$. Moreover, for every constant $s<1$, there is a constant $w>2$ such that every tensor $T$ with $\tilde{S}(T) \leq \tilde{R}(T)^s$ must have $\omega_u(T) \geq w$.
\end{corollary}

Almost all the tensors we consider in this note are \emph{variable-symmetric} tensors, and for these tensors $T$ we can get a better lower bound on $\omega_u(T)$ from an upper bound on $\tilde{S}(T)$. We say that a tensor $T$ over $X,Y,Z$ is variable-symmetric if $|X|=|Y|=|Z|$, and the coefficient of $x_i y_j z_k$ equals the coefficient of $x_j y_k z_i$ in $T$ for all $(x_i, y_j, z_k) \in X \times Y \times Z$.

\begin{theorem}\label{thm:symmSOmega}
For a variable-symmetric tensor $T$ we have $\omega_u(T) \geq 2 \log (\tilde{R}(T)) / \log(\tilde{S}(T)) $.
\end{theorem}

\begin{proof}
As in the proof of~\cite[Theorem~4.1]{aw}, by definition of $\omega_u$, we know that for every $\delta>0$, there is a positive integer $n$ such that $T^{\otimes n}$ has a degeneration to $F \odot \langle a,b,c \rangle$ for integers $F,a,b,c$ such that $\omega_u(T)^{1+\delta} \geq 3 \log(\tilde{R}(T)^n / F) / \log(abc)$. In fact, since $T$ is symmetric, we know $T^{\otimes n}$ also has a degeneration to $F \odot \langle b,c,a \rangle$ and to $F \odot \langle c,a,b \rangle$, and so $T^{\otimes 3n}$ has a degeneration to $F^3 \odot \langle abc, abc, abc \rangle$. As above, it follows that $\tilde{S}(T^{\otimes 3n}) \geq \tilde{S}(F^3 \odot \langle abc, abc, abc \rangle) = F^3 \cdot (abc)^2$. Rearranging, we see
$$abc \leq \tilde{S}(T)^{3n/2} / F^{3/2}.$$
Hence,
$$\omega_u(T)^{1+\delta} \geq 3 \frac{\log(\tilde{R}(T)^n / F) }{ \log(abc)} \geq 3 \frac{\log(\tilde{R}(T)^n / F) }{ \log(\tilde{S}(T)^{3n/2} / F^{3/2})} = 2 \frac{\log(\tilde{R}(T)) - \frac{1}{n} \log(F) }{ \log(\tilde{S}(T)) - \frac{1}{n}\log(F)} \geq 2 \frac{\log(\tilde{R}(T)) }{ \log(\tilde{S}(T)) },$$
where the last step follows because $\tilde{R}(T) \geq \tilde{S}(T)$ and so subtracting the same quantity from both the numerator and denominator cannot decrease the value of the fraction. This holds for all $\delta>0$ and hence implies the desired result.
\end{proof}

\paragraph{Slice Rank versus Subrank} For a tensor $T$, let $Q'(T)$ denote the largest integer $q$ such that there is a degeneration from $T$ to $\langle q \rangle$. The \emph{asymptotic subrank} of $T$ is defined as $\tilde{Q}(T) := \limsup_{n \in \N} Q'(T^{\otimes n})^{1/n}$. Propositions~\ref{propone} and \ref{proptwo} above imply that $\tilde{Q}(T) \leq \tilde{S}(T)$ for all tensors $T$. Similarly, it is not hard to see that Theorems~\ref{thm:SROmega} and \ref{thm:symmSOmega} hold with $\tilde{S}$ replaced by $\tilde{Q}$. One could thus conceivably hope to prove stronger lower bounds than those in this paper by bounding $\tilde{Q}$ instead of $\tilde{S}$. However, we will prove in Corollary~\ref{subrankcorr} below that $\tilde{Q}(T) = \tilde{S}(T)$ for every tensor we study in this paper, so such an improvement using $\tilde{Q}$ is impossible. More generally, there are currently no known tensors $T$ for which the best known upper bound on $\tilde{Q}(T)$ is smaller than the best known upper bound on $\tilde{S}(T)$ (including the new bounds of~\cite{quantum,theirpreprint}). Hence, novel tools for upper bounding $\tilde{Q}$ would be required for such an approach to proving better lower bounds on $\omega_u$.

\subsection{Partition Notation} \label{partitionnotation}

In a number of our results, we will be partitioning the terms of tensors into blocks defined by partitions of the three variable sets. Here we introduce some notation for some properties of such partitions; these definitions all depend on the particular partition of the variables being used, which will be clear from context.

Suppose $T$ is a tensor minimal over $X,Y,Z$, and let $X = X_1 \cup \cdots \cup X_{k_X}$, $Y = Y_1 \cup \cdots \cup Y_{k_Y}$, $Z = Z_1 \cup \cdots \cup Z_{k_Z}$ be partitions of the three variable sets. For $(i,j,k) \in [k_X]\times[k_Y]\times[k_Z]$, let $T_{ijk}$ be $T$ restricted to $X_i, Y_j, Z_k$ (i.e. $T$ with $X \setminus X_i$, $Y \setminus Y_j$, and $Z \setminus Z_k$ zeroed out), and let $L = \{ T_{ijk} \mid (i,j,k) \in [k_X]\times[k_Y]\times[k_Z], T_{ijk} \neq 0\}$. $T_{ijk}$ is called a \emph{block} of $T$. For $i \in [k_X]$ let $L_{X_i} = \{ T_{ij'k'} \in L \mid (j',k') \in  [k_Y] \times [k_Z]\}$, and define similarly $L_{Y_j}$ and $L_{Z_k}$.

We will be particularly interested in probability distributions $p : L \to [0,1]$. Let $P(L)$ be the set of such distributions. For such a $p \in P(L)$, and for $i \in [k_X]$, let $p(X_i) := \sum_{T_{ijk} \in L_{X_i}} p(T_{ijk})$, and similarly $p(Y_j)$ and $p(Z_k)$. Then, define $p_X \in \R$ by

$$p_X := \prod_{i \in [k_X]} \left( \frac{|X_i|}{p(X_i)}\right)^{p(X_i)},$$
and $p_Y$ and $p_Z$ similarly. This expression, which arises naturally in the Laser Method, will play an important role in our upper bounds and lower bounds.

\subsection{Tensor Rotations and Variable-Symmetric Tensors} \label{symmetricnotation}

If $T$ is a tensor over $X,Y,Z$, then the \emph{rotation of $T$}, denoted $rot(T)$, is the tensor over $Y,Z,X$ such that for any $(x_i, y_j, z_k) \in X \times Y \times Z$, the coefficient of $x_i y_j z_k$ in $T$ is equal to the coefficient of $y_j z_k x_i$ in $rot(T)$. Tensor $T$ is \emph{variable-symmetric} if $T \simeq rot(T)$.

If $T$ is a variable-symmetric tensor minimal over $X,Y,Z$, then partitions $X = X_1 \cup \cdots \cup X_{k_X}$, $Y = Y_1 \cup \cdots \cup Y_{k_Y}$, $Z = Z_1 \cup \cdots \cup Z_{k_Z}$ of the variable sets are called \emph{$T$-symmetric} if (using the notation of the previous subsection) $k_X = k_Y = k_Z$, $|X_i| = |Y_i| = |Z_i|$ for all $i \in [k_X]$, and the block $T_{jki} \simeq rot(T_{ijk})$ for all $(i,j,k) \in [k_X]^3$. For the $L$ resulting from such a $T$-symmetric partition, a probability distribution $p \in P(L)$ is called \emph{$T$-symmetric} if it satisfies $p(T_{ijk}) = p(T_{jki})$ for all $(i,j,k) \in [k_X]^3$, and we write $P^{sym}(L) \subseteq P(L)$ for the set of such $T$-symmetric distributions. Notice in particular that any $p \in P^{sym}(L)$ satisfies $p_X = p_Y = p_Z$.

\section{Combinatorial Tools for Asymptotic Slice Rank Upper Bounds} \label{secthree}

We now give three general tools for proving upper bounds on $\tilde{S}(T)$ for many tensors $T$. Each of our tools generalizes one of the three main tools of~\cite{aw}, which were bounding the weaker notion $\tilde{I}$ instead of $\tilde{S}$, and could also only apply to a more restrictive set of tensors. We will make clear what previous result we are generalizing, although our presentation here is entirely self-contained.

\subsection{Generalization of \cite[Theorem~5.3]{aw}}

We know that tensors $T$ without many of one variable have small $\tilde{S}(T)$. We begin by showing that if $T$ can be written as a sum of a few tensors, each of which does not have many of one variable, then we can still prove an upper bound on $\tilde{S}(T)$. 

If $X,Y,Z$ are minimal for $T$, then the \emph{measure} of $T$, denoted $\mu(T)$, is given by $\mu(T) := |X| \cdot |Y| \cdot |Z|$. We state two simple facts about $\mu$:

\begin{fact}
For tensors $A$ and $B$,
\begin{itemize}
    \item $\mu(A \otimes B) = \mu(A) \cdot \mu(B)$, and
    \item if $A$ is minimal over $X,Y,Z$, then $\sr(A) \leq \min\{ |X|, |Y|, |Z| \} \leq \mu(A)^{1/3}$.
\end{itemize} 
\end{fact}

\begin{theorem} \label{one}
Suppose $T$ is a tensor, and $T_1, \ldots, T_k$ are tensors with $T = T_1 + \cdots + T_k$. Then, $\tilde{S}(T) \leq \sum_{i=1}^k (\mu(T_i))^{1/3}$.
\end{theorem}

\begin{proof}
Note that
$$T^{\otimes n} = \sum_{(P_1, \ldots, P_n) \in \{ T_1, \ldots, T_k\}^n} P_1 \otimes \cdots \otimes P_n.$$
It follows that
\begin{align*}\sr(T^{\otimes n}) &\leq \sum_{(P_1, \ldots, P_n) \in \{ T_1, \ldots, T_k\}^n} \sr(P_1 \otimes \cdots \otimes P_n) \\&\leq \sum_{(P_1, \ldots, P_n) \in \{ T_1, \ldots, T_k\}^n} \mu(P_1 \otimes \cdots \otimes P_n)^{1/3} \\&= \sum_{(P_1, \ldots, P_n) \in \{ T_1, \ldots, T_k\}^n} (\mu(P_1) \cdot \mu(P_2) \cdots \mu(P_n))^{1/3} \\&= (\mu(T_1)^{1/3} + \cdots + \mu(T_k)^{1/3})^{n},\end{align*}
which implies as desired that $\tilde{S}(T) \leq (\mu(T_1)^{1/3} + \cdots + \mu(T_k)^{1/3})$.
\end{proof}

\begin{remark}
\cite[Theorem~5.3]{aw}, in addition to bounding $\tilde{I}$ instead of $\tilde{S}$, also required that $T = T_1 + \cdots + T_k$ be a \emph{partition} of the terms of $T$. Here in Theorem~\ref{one} we are allowed any tensor sum, although in general a partition minimizes the resulting upper bound.
\end{remark}

\subsection{Generalization of \cite[Theorem~5.2]{aw}} \label{subsectwo}

This tool will be the most important in upper bounding the asymptotic slice rank of many tensors of interest. We show that a partitioning method similar to the Laser Method applied to a tensor $T$ can be used to prove upper bounds on $\tilde{S}(T)$. Recall the definitions and notation about partitions of tensors from Section~\ref{partitionnotation}.

\begin{theorem} \label{two} For any tensor $T$ and partition of its variable sets,
$$\tilde{S}(T) \leq \limsup_{p \in P(L)} \min\{ p_X, p_Y, p_Z \}.$$
\end{theorem}

\begin{proof}
For any positive integer $n$, we can write
$$T^{\otimes n} = \sum_{(P_1, \ldots, P_n) \in L^n} P_1 \otimes \cdots \otimes P_n.$$
For a given $(P_1, \ldots, P_n) \in L^n$, let $dist(P_1, \ldots, P_n)$ be the probability distribution on $L$ which results from picking a uniformly random $\alpha \in [n]$ and outputting $P_\alpha$. For a probability distribution $p : L \to [0,1]$, define $L_{n,p} := \{ (P_1, \ldots, P_n) \in L^n \mid dist(P_1, \ldots, P_n) = p  \}$. Note that the number of $p$ for which $L_{n,p}$ is nonempty is only $\poly(n)$, since they are the distributions which assign an integer multiple of $1/n$ to each element of $L$. Let $D$ be the set of these probability distributions.

We can now rearrange:
$$T^{\otimes n} = \sum_{p \in D} \sum_{(P_1, \ldots, P_n) \in L_{n,p}}P_1 \otimes \cdots \otimes P_n.$$

Hence,
\begin{align*}
    \sr(T^{\otimes n}) &\leq \sum_{p \in D} \sr\left( \sum_{(P_1, \ldots, P_n) \in L_{n,p}}P_1 \otimes \cdots \otimes P_n \right) \\ &\leq \poly(n) \cdot \max_{p \in D} \sr\left( \sum_{(P_1, \ldots, P_n) \in L_{n,p}}P_1 \otimes \cdots \otimes P_n \right).
\end{align*}

For any probability distribution $p : L \to [0,1]$, let us count the number of x-variables used in $\left( \sum_{(P_1, \ldots, P_n) \in L_{n,p}}P_1 \otimes \cdots \otimes P_n \right)$. These are the tuples of the form $(x_1, \ldots, x_n) \in X^n$ where, for each $i \in [k_X]$, there are exactly $n \cdot p(X_i)$ choices of $j$ for which $x_j \in X_i$. The number of these is\footnote{Here, $\binom{n}{p_1 n, p_2 n, \ldots, p_\ell n } = \frac{n!}{(p_1 n)! (p_2 n)! \cdots (p_\ell n)!}$, with each $p_i \in [0,1]$ and $p_1 + \cdots + p_\ell = 1$, is the multinomial coefficient, with the known bound from Stirling's approximation, for fixed $p_i$s, that $\binom{n}{p_1 n, p_2 n, \ldots, p_\ell n} \leq  \left( \prod_i p_i^{- p_i} \right)^{n + o(n)} $. Throughout this paper we use the convention that $p_i^{p_i} = 1$ when $p_i = 0$.}
$$\binom{n}{n \cdot p(X_1), n \cdot p(X_2), \ldots, n \cdot p(X_{k_X})} \cdot \prod_{i \in [k_X]} |X_i|^{n \cdot p(X_i)}.$$
This is upper bounded by $p_X^{n + o(n)}$, where $p_X$ is the quantity defined in Section~\ref{partitionnotation}. It follows that $\xr \left( \sum_{(P_1, \ldots, P_n) \in L_{n,p}}P_1 \otimes \cdots \otimes P_n \right) \leq p_X^{n + o(n)}$. We can similarly argue about $\yr$ and $\zr$. Hence,
\begin{align*}
    \sr(T^{\otimes n}) &\leq \poly(n) \cdot \max_{p \in D} \sr\left( \sum_{(P_1, \ldots, P_n) \in L_{n,p}}P_1 \otimes \cdots \otimes P_n \right) \\ &\leq \poly(n) \cdot \max_{p \in D} \min\{ p_X, p_Y, p_Z\}^{n + o(n)} \\ &\leq \poly(n) \cdot \limsup_{p \in P(L)} \min\{ p_X, p_Y, p_Z\}^{n + o(n)}.
\end{align*}
Hence, $\sr(T^{\otimes n}) \leq \limsup_{p} \min\{ p_X, p_Y, p_Z\}^{n + o(n)}$, and the desired result follows.
\end{proof}

\begin{remark}
\cite[Theorem~5.2]{aw} is less general than our Theorem~\ref{two} in two ways: it used $\tilde{I}$ instead of $\tilde{S}$, and it required each $X_i, Y_j, Z_k$ to contain only one variable.
\end{remark}

\begin{remark}
Suppose $T$ is over $X,Y,Z$ with $|X|=|Y|=|Z|=q$. For any probability distribution $p$ we always have $p_X, p_Y, p_Z \leq q$, and moreover we only have $p_X = q$ when $p(X_i) = |X_i|/q$ for each $i$. Similar to~\cite[Corollary~5.1]{aw}, it follows that if no probability distribution $p$ is $\delta$-close (say, in $\ell_1$ distance) to having $p(X_i) = |X_i|/q$ for all $i$, $p(Y_j) = |Y_j|/q$ for all $j$, and $p(Z_k) = |Z_k|/q$ for all $k$, simultaneously, then we get $\tilde{S}(T) \leq q^{1 - f(\delta)}$ for some increasing function $f$ with $f(\delta)>0$ for all $\delta>0$. 
\end{remark}

We make a remark about applying Theorem~\ref{two} to variable-symmetric tensors. This remark has implicitly been used in past work on applying the Laser method, such as~\cite{coppersmith}, but we prove it here for completeness. Recall the notation in Section~\ref{symmetricnotation} about such tensors.

\begin{proposition} \label{symmprop}
Suppose $T$ is a variable-symmetric tensor over $X,Y,Z$, and $X = X_1 \cup \cdots \cup X_{k_X}$, $Y = Y_1 \cup \cdots \cup Y_{k_Y}$, $Z = Z_1 \cup \cdots \cup Z_{k_Z}$ are $T$-symmetric partitions. Then,
$$\tilde{S}(T) \leq \limsup_{p \in P^{sym}(L)} p_X.$$
\end{proposition}

\begin{proof}
We know from Theorem~\ref{two} that $\tilde{S}(T) \leq \limsup_{p \in P(L)} \min\{ p_X, p_Y, p_Z \}$. We will show that for any $p \in P(L)$, there is a $p' \in P^{sym}(L)$ such that $\min\{ p_X, p_Y, p_Z \} \leq \min\{ p'_X, p'_Y, p'_Z \}$, which means that in fact, $\tilde{S}(T) \leq \limsup_{p \in P^{sym}(L)} \min\{ p_X, p_Y, p_Z \}$. Finally, the desired result will follow since, for any $p' \in P^{sym}(L)$, we have $p'_X = p'_Y = p'_Z$. 

Consider any $p \in P(L)$, and define the distribution $p' \in P^{sym}(L)$ by $p'(T_{ijk}) := (p(T_{ijk}) + p(T_{jki}) + p(T_{kij}))/3$ for each $T_{ijk} \in L$. In order to show that $\min\{ p_X, p_Y, p_Z \} \leq p'_X$, we will show that $(p_X p_Y p_Z)^{1/3} \leq p'_X$:
\begin{align*} 
(p_X p_Y p_Z)^{1/3} &=  \prod_{i \in [k_X]} \left( \frac{|X_i|}{p(X_i)}\right)^{p(X_i)/3} \left( \frac{|Y_i|}{p(Y_i)}\right)^{p(Y_i)/3} \left( \frac{|Z_i|}{p(Z_i)}\right)^{p(Z_i)/3}  \\
&= \prod_{i \in [k_X]} \frac{|X_i|^{p'(X_i)}}{(p(X_i)^{p(X_i)} p(Y_i)^{p(Y_i)} p(Z_i)^{p(Z_i)})^{1/3}} \\
&\leq \prod_{i \in [k_X]} \frac{|X_i|^{p'(X_i)}}{p'(X_i)^{p'(X_i)}} \\ 
&= p'_X,
\end{align*}
where the second-to-last step follows from the fact that for any real numbers $a,b,c \in [0,1]$, setting $d = (a+b+c)/3$, we have $a^a b^b c^c \geq d^{3d}$.
\end{proof}

\subsection{Generalization of \cite[Theorem~5.1]{aw}}
The final remaining tool from \cite{aw}, their Theorem~5.1, turns out to be unnecessary for proving our tight lower bounds in the next section. Nonetheless, we sketch here how to extend it to give asymptotic slice rank upper bounds as well.

For a tensor $T$, let $m(T) := \max \{ \xr(T), \yr(T), \zr(T) \}$. Recall from Lemma~\ref{propertieslemma} that for any two tensors $A,B$ we have $\sr(A \otimes B) \leq \sr(A) \cdot m(B)$.

In general, for two tensors $A$ and $B$, even if $\tilde{S}(A)$ and  $\tilde{S}(B)$ are `small', it might still be the case that $\tilde{S}(A+B)$ is `large', much larger than $\tilde{S}(A) + \tilde{S}(B)$. For instance, for any positive integer $q$, define the tensors $T_1 := \sum_{i=0}^q x_0 y_i z_i$, $T_2 := \sum_{i=1}^{q+1} x_i y_0 z_i$, and $T_3 := \sum_{i=1}^{q+1} x_i y_i z_{q+1}$. We can see that $\tilde{S}(T_1) = \tilde{S}(T_2) = \tilde{S}(T_3) = 1$, but $T_1 + T_2 + T_3 = CW_{q}$, and we will show soon that $\tilde{S}(CW_q)$ grow unboundedly with $q$.

Here we show that if, not only is $\tilde{S}(A)$ small, but even $\xr(A)$ is small, then we can get a decent bound on  $\tilde{S}(A+B)$.

\begin{theorem} \label{removeanxtool}
Suppose $T,A,B$ are tensors such that $A + B = T$. Then,
$$\tilde{S}(T) \leq \left( \frac{m(A)}{(1-p) \cdot \xr(A)} \right)^{1-p} \cdot \frac{1}{p^p},$$
where $p \in [0,1]$ is given by
$$p := \frac{\log\left( \frac{\xr(B)}{\tilde{S}(B)} \right)}{ \log\left(\frac{m(A)}{\xr(A)}\right) + \log\left(\frac{\xr(B)}{\tilde{S}(B)}\right)}.$$
\end{theorem}

\begin{proof}
We begin by, for any integers $n \geq k \geq 0$, giving bounds on $\sr(A^{\otimes k} \otimes B^{\otimes (n-k)})$. First, since $\xr$ is submultiplicative, we have $$\sr(A^{\otimes k} \otimes B^{\otimes (n-k)}) \leq \xr(A^{\otimes k} \otimes B^{\otimes (n-k)}) \leq \xr(A)^k \cdot \xr(B)^{n-k}.$$
Second, from the definition of $m$, we have $$\sr(A^{\otimes k} \otimes B^{\otimes (n-k)}) \leq m(A^{\otimes k}) \cdot \sr(B^{\otimes (n-k)})  \leq m(A)^k \cdot \tilde{S}(B)^{n-k}.$$

It follows that for any positive integer $n$ we have
$$S(T^{\otimes n}) \leq \sum_{k=0}^n \binom{n}{k} \cdot \sr(A^{\otimes k} \otimes B^{\otimes (n-k)}) \leq \sum_{k=0}^n \binom{n}{k} \cdot \min \{ \xr(A)^k \cdot \xr(B)^{n-k}, m(A)^k \cdot \tilde{S}(B)^{n-k} \}.$$

As in the proof of~\cite[Theorem~5.1]{aw}, we can see that the quantity $\binom{n}{k} \cdot \min \{ \xr(A)^k \cdot \xr(B)^{n-k}, m(A)^k \cdot \tilde{S}(B)^{n-k} \}$ is maximized at $k = pn$, and the result follows.
\end{proof}

\begin{remark}
This result generalizes \cite[Theorem~5.1]{aw}, no longer requiring that $A$ be the tensor $T$ restricted to a single $x$-variable. In \cite[Theorem~5.1]{aw}, since $A$ is $T$ restricted to a  single $x$-variable, and we required $A$ to have at most $q$ terms, we got the bounds $\xr(A) = 1$ and $m(A) \leq q$. Similarly, $B$ had at most $q-1$ different $x$-variables, so $\xr(B) \leq q-1$. Substituting those values into Theorem~\ref{removeanxtool} yields the original~\cite[Theorem~5.1]{aw} with $\tilde{I}$ replaced by $\tilde{S}$. 
\end{remark}

\section{Computing the Slice Ranks for Tensors of Interest} \label{secfour}

In this section, we give slice rank upper bounds for a number of tensors of interest. It will follow from Section~\ref{lasersec} that \emph{all of the bounds we prove in this Section are tight}.

\subsection{Generalized Coppersmith-Winograd Tensors} \label{secgcw}

We begin with the generalized CW tensors defined in~\cite{aw}, which for a positive integer $q$ and a permutation $\sigma : [q] \to [q]$ are given by
$$CW_{q,\sigma} := x_0 y_0 z_{q+1} + x_0 y_{q+1} z_0 + x_{q+1} y_0 z_0 + \sum_{i=1}^q (x_i y_{\sigma(i)} z_0 + x_i y_0 z_i + x_0 y_i z_i).$$

The usual Coppersmith-Winograd tensor $CW_q$ results by setting $\sigma$ to the identity permutation. Just as in \cite[Section~7.1]{aw}, we can see that Theorems~\ref{one}~and~\ref{two} immediately apply to $CW_{q,\sigma}$ to show that there is a universal constant $\delta>0$ such that for any $q$ and $\sigma$ we have $\tilde{S}(CW_{q,\sigma}) \leq (q+2)^{1-\delta}$, and hence a universal constant $c>2$ such that $\omega_u(CW_{q,\sigma}) \geq c$. Indeed, by proceeding in this way, we get the exact same constants as in~\cite{aw}.

That said, we will now use Theorem~\ref{two} to prove that $c \geq 2.16805$. (In fact, essentially the same argument as we present now shows that \cite[Theorem 5.2]{aw} was already sufficient to show the weaker claim that $\omega_g(CW_{q,\sigma}) \geq 2.16805$).

We begin by partitioning the variable sets of $CW_{q,\sigma}$, using the notation of Theorem~\ref{two}. Let $X_0 = \{ x_0 \}$, $X_1 = \{ x_1, \ldots, x_q \}$, and $X_2 = \{ x_{q+1} \}$, so that $X_0 \cup X_1 \cup X_2$ is a partition of the $x$-variables of $CW_{q,\sigma}$.\footnote{The sets of partitions were 1-indexed before, but we 0-index here for notational consistency with past work.} Similarly, let $Y_0 = \{ y_0 \}$, $Y_1 = \{ y_1, \ldots, y_q \}$, $Y_2 = \{ y_{q+1} \}$, $Z_0 = \{ z_0 \}$, $Z_1 = \{ z_1, \ldots, z_q \}$, and $Z_2 = \{ z_{q+1} \}$. We can see this is a $CW_{q,\sigma}$-symmetric partition with $L = \{ T_{002}, T_{020}, T_{200}, T_{011}, T_{101}, T_{110} \}$. 

Consider any probability distribution $p \in P^{sym}(L)$. By symmetry, we know that $p(T_{002}) = p(T_{020}) = p(T_{200}) = v$ and $p(T_{011}) = p(T_{101}) = p(T_{110}) = 1/3 - v$ for some value $v \in [0,1/3]$. Applying Theorem~\ref{two}, and in particular Proposition~\ref{symmprop}, yields:

$$\tilde{S}(CW_q) \leq \sup_{v \in [0,1/3]} \frac{q^{2(1/3 - v)}}{v^v (2/3 - 2v)^{2/3 - 2v} (1/3 + v)^{1/3+v}}.$$

In fact, we will see in the next section that this is tight (i.e. the value above is \emph{equal} to $\tilde{S}(CW_q)$, not just an upper bound on it). The values for the first few $q$ can be computed using optimization software as follows:

\begin{center}
\begin{tabular}{l|l}
$q$ & $\tilde{S}(CW_{q,\sigma})$ \\ \hline
$1$                                & $2.7551\cdots$                                                        \\
$2$                                & $3.57165\cdots$                                                        \\
$3$                                & $4.34413\cdots$                                                        \\
$4$                                & $5.07744\cdots$                                                        \\
$5$                                & $5.77629\cdots$                                                           \\
$6$                                & $6.44493\cdots$                                                       \\
$7$                                & $7.08706\cdots$                                                       \\
$8$                                & $7.70581\cdots$                                                   
\end{tabular}
\end{center}

Finally, using the lower bound $\tilde{R}(CW_{q,\sigma}) \geq q+2$ (in fact, it is known that $\tilde{R}(CW_{q,\sigma}) = q+2$), and the upper bound on $\tilde{S}(CW_{q,\sigma})$ we just proved, we can apply Theorem~\ref{thm:symmSOmega} to give lower bounds $\omega_u(CW_{q,\sigma}) \geq 2 \log(\tilde{R}(CW_{q,\sigma})) / \log(\tilde{S}(CW_{q,\sigma})) \geq 2 \log(q+2) / \log(\tilde{S}(CW_{q,\sigma}))$ as follows:

\begin{center}
\begin{tabular}{l|l}
$q$ & Lower Bound on $\omega_u(CW_{q,\sigma})$ \\ \hline
$1$                                & $2.16805\cdots$                                                        \\
$2$                                & $2.17794\cdots$                                                        \\
$3$                                & $2.19146\cdots$                                                        \\
$4$                                & $2.20550\cdots$                                                        \\
$5$                                & $2.21912\cdots$                                                           \\
$6$                                & $2.23200\cdots$                                                       \\
$7$                                & $2.24404\cdots$                                                       \\
$8$                                & $2.25525\cdots$                                                        
\end{tabular}
\end{center}

It is not hard to see that the resulting lower bound on $\omega_u(CW_{q,\sigma})$ is increasing with $q$ and is always at least $2.16805\ldots$ (see Appendix~\ref{appone} below for a proof), and hence that for any $q$ and any $\sigma$ we have $\omega_u(CW_{q,\sigma}) \geq 2.16805$ as desired.

\subsection{Generalized Simple Coppersmith-Winograd Tensors}

Similar to $CW_{q,\sigma}$, we can define for a positive integer $q$ and a permutation $\sigma : [q] \to [q]$ the simple Coppersmith-Winograd tensor $cw_{q,\sigma}$ given by:
$$cw_{q,\sigma} := \sum_{i=1}^q (x_i y_{\sigma(i)} z_0 + x_i y_0 z_i + x_0 y_i z_i).$$

These tensors, when $\sigma$ is the identity permutation $id$, are well-studied. For instance, Coppersmith and Winograd~\cite{coppersmith} showed that if $\tilde{R}(cw_{2,id}) = 2$ then $\omega=2$.

We will again give a tight bound on $\tilde{S}(cw_{q,\sigma})$ using Theorem~\ref{two} combined with the next section. To apply Theorem~\ref{two}, and in particular Proposition~\ref{symmprop}, we again pick a partition of the variables. Let $X_0 = \{ x_0 \}$, $X_1 = \{ x_1, \ldots, x_q \}$, $Y_0 = \{ y_0 \}$, $Y_1 = \{ y_1, \ldots, y_q \}$,  $Z_0 = \{ z_0 \}$, and $Z_1 = \{ z_1, \ldots, z_q \}$. This is a $cw_{q,\sigma}$-symmetric partition with $L = \{ T_{011}, T_{101}, T_{110} \}$. There is a unique $p \in P^{sym}(L)$, which assigns probability $1/3$ to each part. It follows that

$$\tilde{S}(cw_{q,\sigma}) \leq (1/3)^{-1/3} (2/3)^{-2/3} \cdot q^{2/3} = \frac{3}{2^{2/3}} \cdot q^{2/3}.$$

Again, we will see in the next section that this bound is tight. Using the lower bound $\tilde{R}(cw_{q,\sigma}) \geq q+1$, we get the lower bound
$$\omega_u(cw_{q,\sigma}) \geq 2 \frac{\log(q+1) }{ \log\left(\frac{3}{2^{2/3}} \cdot q^{2/3}\right)}.$$

The first few values are as follows; note that we cannot get a bound better than $2$ when $q=2$ because of Coppersmith and Winograd's remark.

\begin{center}
\begin{tabular}{l|l}
$q$ & Lower Bound on $\omega_u(cw_{q,\sigma})$ \\ \hline
$1$                                & $2.17795\cdots$                                                        \\
$2$                                & $2$                                                        \\
$3$                                & $2.02538\cdots$                                                        \\
$4$                                & $2.06244\cdots$                                                        \\
$5$                                & $2.09627\cdots$                                                           \\
$6$                                & $2.12549\cdots$                                                       \\
$7$                                & $2.15064\cdots$                                                   
\end{tabular}
\end{center}

\subsection{Cyclic Group Tensors}

We next look at two tensors which were studied in \cite{cohn2003group}, \cite{almanitcs}, and \cite[Section~7.3]{aw}. For each positive integer $q$, define the tensor $T_q$ (the structural tensor of the cyclic group $C_q$) as:
$$T_q=\sum_{i=0}^{q-1}\sum_{j=0}^{q-1} x_i y_j z_{i+j\bmod q}.$$
Define also the lower triangular version of $T_q$, called $T_q^{lower}$, as:
$$T_q^{lower}=\sum_{i=0}^{q-1}\sum_{j=0}^{q-1-i} x_i y_j z_{i+j}.$$

 While Theorem~\ref{two} does not give any nontrivial upper bounds on $\tilde{S}(T_q)$, it does give nontrivial upper bounds on $\tilde{S}(T_q^{lower})$, as noted in \cite[Section~7.3]{aw}. Using computer optimization software, we can compute our lower bound on $\tilde{S}(T_q^{lower})$, using Theorem~\ref{two} where each partition contains exactly one variable, for the first few values of $q$:

\begin{center}
\begin{tabular}{l|l}
$q$ & Upper Bound on $\tilde{S}(T_q^{lower})$ \\ \hline
$2$                                & $1.88988\cdots$                                                        \\
$3$                                & $2.75510\cdots$                                                        \\
$4$                                & $3.61071\cdots$                                                        \\
$5$                                & $4.46157\cdots$                                                     
\end{tabular}
\end{center}

We show in the next section that these numbers are also tight. It is known (see e.g.~\cite{almanitcs}) that $\tilde{R}(T_q) = \tilde{R}(T_q^{lower}) = q$. Thus we get the following lower bounds on $\omega_u(T_q^{lower}) \geq 2 \log(q) / \log(\tilde{S}(T_q^{lower}))$:

\begin{center}
\begin{tabular}{l|l}
$q$ & Lower Bound on $\omega_u(T_q^{lower})$ \\ \hline
$2$                                & $2.17795\cdots$                                                        \\
$3$                                & $2.16805\cdots$                                                        \\
$4$                                & $2.15949\cdots$                                                        \\
$5$                                & $2.15237\cdots$                                                     
\end{tabular}
\end{center}

These numbers match the lower bounds obtained by \cite{almanitcs, blasiak} in their study of $T_q$; our Theorem~\ref{two} can be viewed as an alternate tool to achieve those lower bounds. The bound approaches $2$ as $q \to \infty$, as it is known that $\log(\sr(T_q))/\log(q) = 1 - o(1)$ as $q \to \infty$.  Interestingly, it is shown in \cite[Theorem~4.16]{quantum} that $T_q^{lower}$ degenerates to $T_q$ over the field $\F_q$, which implies that our bounds above also hold for $T_q$ over $\F_q$.

\subsection{The Value of the Subtensor $t_{112}$ of $CW_q^{\otimes 2}$} \label{t112}

A key tensor which arises in applying the Laser method to increasing powers of $CW_q$, including \cite{coppersmith, v12,legall,legallrect, legallrect2}, is the tensor $t_{112}$ which (for a given positive integer $q$) is given by

$$t_{112} := \sum_{i=1}^q x_{i,0} y_{i,0} z_{0,q+1} + \sum_{k=1}^q x_{0,k} y_{0,k} z_{q+1,0} + \sum_{i,k=1}^q x_{i,0} y_{0,k} z_{i,k} + \sum_{i,k = 1}^q x_{0,k} y_{i,0} z_{i,k}.$$

Coppersmith-Winograd~\cite{coppersmith} and future work studied the value of this tensor. In~\cite{coppersmith} it is shown that for every $\tau \in [2/3, 1]$,

$$V_\tau(t_{112}) \geq 2^{2/3} q^\tau (q^{3\tau} + 2)^{1/3}.$$

This bound has been used in all the subsequent work using $CW_q$, without improvement. Here we show it is tight and cannot be improved in the case $\tau=2/3$:

\begin{proposition} \label{valueprop}
$V_{2/3}(t_{112}) = 2^{2/3} q^{2/3} (q^{2} + 2)^{1/3}.$
\end{proposition}

\begin{proof}
Consider the variable-symmetric tensor $t_s := t_{112} \otimes rot(t_{112}) \otimes rot(rot(t_{112}))$.
As in~\cite{coppersmith}, by definition of $V_{2/3}$, for every $\delta>0$ there is a positive integer $n$ such that $t_s^{\otimes n}$ has a degeneration to $\bigoplus_i \langle a_i, a_i, a_i \rangle$ for values such that $\sum_i a_i^2 \geq (V_{2/3}(T_{112}))^{3n(1-\delta)}$. In particular, by Corollary~\ref{cortopropthree} this yields the bound $$\tilde{S}(t_s^{\otimes n}) \geq \sum_i a_i^2 \geq (V_{2/3}(t_{112}))^{3n(1-\delta)}.$$ Since this holds for all $\delta>0$, it follows that $\tilde{S}(t_s)  \geq (V_{2/3}(t_{112}))^3 \geq 2^{2} q^{2} (q^{2} + 2)$.

We now upper bound $\tilde{S}(t_s)$ using Theorem~\ref{two}. Although we are analyzing $t_s$, we will make use of a partition of the variables of $t_{112}$. The partition is as follows: $X_0 = \{ x_{i,0} \mid i \in [q] \}$, $X_1 = \{ x_{0,k} \mid k \in [q] \}$, $Y_0 = \{ y_{i,0} \mid i \in [q] \}$, $Y_1 = \{ y_{0,k} \mid k \in [q] \}$, $Z_0 = \{ z_{i,k} \mid i,k \in [q] \}$, $Z_1 = \{ z_{0,q+1} \}$, and $Z_2 = \{ z_{q+1,0} \}$. Hence, $L = \{ T_{001}, T_{112}, T_{010}, T_{100} \}$. As in \cite{coppersmith}, and similar to Proposition~\ref{symmprop}, since $t_s$ is defined as $t_s := t_{112} \otimes rot(t_{112}) \otimes rot(rot(t_{112}))$, it follows that $\tilde{S}(t_s) \leq \limsup_{p \in P(L)} p_X \cdot p_Y \cdot p_Z$.
We can assume, again by symmetry, that any probability distribution $p$ on $L$ assigns the same value $v$ to $T_{010}$ and $T_{100}$, and the same value $1/2 - v$ to $T_{001}$ and $T_{112}$. We finally get the bound:

$$\tilde{S}(t_s) \leq \limsup_{v \in [0,1/2]} \left( 2q \right)^2 \cdot \frac{(q^2)^{2v}}{(2v)^{2v} (1/2-v)^{1-2v}}.$$

This is maximized at $v = q^2 / (2q^2+2)$, which yields exactly $\tilde{S}(t_s) \leq 2^{2} q^{2} (q^{2} + 2)$. The desired bound follows.
\end{proof}

The only upper bound we are able to prove on $V_{\tau}$ for $\tau>2/3$ is the straightforward $V_\tau(t_{112}) \leq V_{2/3}(t_{112})^{3 \tau/2} = 2^{\tau} q^{\tau} (q^{2} + 2)^{\tau/2}$, which is slightly worse than the best known lower bound $V_\tau(t_{112}) \geq 2^{2/3} q^\tau (q^{3\tau} + 2)^{1/3}$. It is an interesting open problem to prove tight upper bounds on $V_\tau(T)$ for any nontrivial tensor $T$ and value $\tau>2/3$. $T = t_{112}$ may be a good candidate since the Laser method seems unable to improve $V_\tau(t_{112})$ for any $\tau$, even when applied to any small tensor power $t_{112}^{\otimes n}$.

Notice that we were able to prove a tight bound on $\tilde{S}(t_s)$ here: the upper bound we proved matches a lower bound which we were able to derive from Coppersmith-Winograd's analysis (which made use of the Laser Method) of $V_\tau(t_{112})$. In the next section we will substantially generalize this fact, by showing a tight bound on $\tilde{S}(T)$ for any tensor $T$ to which the Laser Method applies.

\section{Slice Rank Lower Bounds via the Laser Method} \label{lasersec}

In this section, we show that the Laser Method can be used to give matching upper and lower bounds on $\tilde{S}(T)$ for any tensor $T$ to which it applies. We will build off of Theorem~\ref{two}, which we will show matches the bounds which arise in the Laser Method.

Consider any tensor $T$ which is minimal over $X,Y,Z$, and let $X = X_1 \cup \cdots \cup X_{k_X}$, $Y = Y_1 \cup \cdots \cup Y_{k_Y}$, $Z = Z_1 \cup \cdots \cup Z_{k_Z}$ be partitions of the three variable sets. Define $T_{ijk}$, $L$, and $p_X$ for a probability distribution $p$ on $L$, as in the top of Subsection~\ref{partitionnotation}. Recall in particular that $T_{ijk}$ is $T$ restricted to the variable sets $X_i$, $Y_j$, and $Z_k$.

\begin{definition} \label{laserable}
We say that $T$, along with partitions of $X,Y,Z$, is a \emph{laser-ready tensor partition} if the following three conditions are satisfied:
\begin{enumerate}[(1)]
    \item For every $(i,j,k) \in [k_X] \times [k_Y] \times [k_Z]$, either $T_{ijk} = 0$, or else $T_{ijk}$ has a degeneration to a tensor $\langle a,b,c \rangle$ with $ab = |X_i|$, $bc = |Y_j|$, and $ca = |Z_k|$ (i.e. a matrix multiplication tensor which is as big as possible given $|X_i|$, $|Y_j|$, and $|Z_k|$).
    \item There is an integer $\ell$ such that $T_{ijk} \neq 0$ only if $i+j+k=\ell$.
    \item $T$ is variable-symmetric, and the partitions are $T$-symmetric.
\end{enumerate}
\end{definition}

These conditions are exactly those for which the original Laser Method used by Coppersmith and Winograd~\cite{coppersmith} applies to $T$. We note that condition (3) is a simplifying assumption rather than a real condition on $T$: for any tensor $T$ and partitions satisfying conditions (1) and (2), the tensor $T' := T \otimes rot(T) \otimes rot(rot(T))$ along with the corresponding product partitions, satisfies all three conditions, gives at least as good a bound on $\omega$ using the Laser Method as $T$ and the original partitions, and more generally has $\omega_u(T') \leq \omega_u(T)$.

\begin{theorem}[\cite{coppersmith,stothers,v12}] \label{laserthm}
Suppose $T$, along with the partitions of $X,Y,Z$, is a laser-ready tensor partition. Then, for any distribution $p\in P^{sym}(L)$, and any positive integer $n$, the tensor $T^{\otimes n}$ has a degeneration into $$\left(\prod_{i \in [k_X]} p(X_i)^{-p(X_i)}\right)^{n - o(n)} \odot \langle a,a,a \rangle,$$
where 
$$a = \left( \prod_{T_{ijk} \in L} |X_i| ^{p(T_{ijk})} \right)^{n/2 - o(n)}.$$
\end{theorem}

\begin{proof}
Typically, as described in \cite[Section~3]{v12}, there is an additional loss in the size of the degeneration if there are multiple different distributions $p, p'$ with the same marginals (meaning $p(X_i) = p'(X_i)$, $p(Y_j) = p'(Y_j)$, and $p(Z_k) = p'(Z_k)$ for all $i,j,k$) but different values of $V(p) := \prod_{T_{ijk} \in L} V_{\tau}(T_{ijk})^{p(T_{ijk})}$ for any $\tau \in [2/3, 1]$. However, because of condition (1) in the definition of a laser-ready tensor partition, the quantity $V(p)$ is equal to 
$$ \prod_{T_{ijk} \in L} (|X_i| \cdot |Y_j| \cdot |Z_k|)^{p(T_{ijk})\cdot \tau /2},$$
and in particular satisfies $V(p) = V(p')$ for any two distributions $p,p'$ with the same marginals. Thus, we do not incur this loss, and we get the desired degeneration. 
\end{proof}

Our key new result about such tensor partitions is as follows:

\begin{theorem}\label{thmlaseropt}
Suppose tensor $T$, along with the partitions of $X,Y,Z$, is a laser-ready tensor partition. Then, $$\tilde{S}(T) = \limsup_{p \in P^{sym}(L)} p_X.$$
\end{theorem}

\begin{proof}
The upper bound, $\tilde{S}(T) \leq \limsup_{p \in P^{sym}(L)} p_X$, is given by Proposition~\ref{symmprop}.

For the lower bound, we know from Theorem~\ref{laserthm} that for all $p \in P^{sym}(L)$, and all positive integers $n$, the tensor $T^{\otimes n}$ has a degeneration into $$\left(\prod_{i \in [k_X]} p(X_i)^{-p(X_i)}\right)^{n - o(n)} \odot \langle a,a,a \rangle,$$
where 
$$a = \left( \prod_{T_{ijk} \in L} |X_i| ^{p(T_{ijk})} \right)^{n/2 - o(n)}.$$
By Proposition~\ref{propthree}, this means $T^{\otimes n}$ has a degeneration to an independent tensor of size
$$\left(\prod_{i \in [k_X]} p(X_i)^{-p(X_i)}\right)^{n - o(n)} \cdot a^2 = p_X^{n - o(n)}.$$
Applying Propositions~\ref{propone}~and~\ref{proptwo} implies that $\tilde{S}(T) \geq p_X$ for all $p \in P^{sym}(L)$, as desired.
\end{proof}

\begin{corollary}
The upper bounds on $\tilde{S}(CW_{q,\sigma})$, $\tilde{S}(cw_{q,\sigma})$, $\tilde{S}(T_q^{lower})$, and $\tilde{S}(T_q)$ from Section~\ref{secfour} are tight.
\end{corollary}

\begin{proof}
$CW_{q,\sigma}$, $cw_{q,\sigma}$, and $T_q^{lower}$, partitioned as they were in the previous section, are laser-ready tensor partitions. The tight bound for $T_q$ follows from the degeneration to $T_q^{lower}$ described in the previous section.
\end{proof}

\begin{corollary} \label{subrankcorr}
Every tensor $T$ with a laser-ready tensor partition (including $CW_{q,\sigma}$, $cw_{q,\sigma}$, and $T_q^{lower}$) has $\tilde{S}(T) = \tilde{Q}(T)$.
\end{corollary}

\begin{proof}
All tensors satisfy $\tilde{S}(T) \geq \tilde{Q}(T)$. In Theorem~\ref{thmlaseropt}, the upper bound on $\tilde{S}(T)$ showed that $T^{\otimes n}$ has a degeneration to an independent tensor of size $\tilde{S}(T)^{n - o(n)}$, which implies that $\tilde{Q}(T) \geq \tilde{S}(T)$.
\end{proof}

\begin{corollary}
If $T$ is a tensor with a laser-ready tensor partition, and applying the Laser method to $T$ with this partition yields an upper bound on $\omega$ of $\omega_u(T) \leq c$ for some $c>2$, then $\omega_u(T) > 2$.
\end{corollary}

\begin{proof}
When the Laser method shows, as in Theorem~\ref{laserthm}, that $T^{\otimes n}$ has a degeneration into $$\left(\prod_{i \in [k_X]} p(X_i)^{-p(X_i)}\right)^{n - o(n)} \odot \langle a,a,a \rangle,$$
the resulting upper bound on $\omega_u(T)$ is that $$\left(\prod_{i \in [k_X]} p(X_i)^{-p(X_i)}\right)^{n - o(n)} \cdot a^{\omega_u(T)} \geq \tilde{R}(T)^n.$$
In particular, since the left-hand side equals $p_X$ when $\omega_u(T)=2$, this yields $\omega_u(T)=2$ if and only if $p_X = \tilde{R}(T)$, so if it yields $\omega_u(T) \leq c$, then $\tilde{S}(T) = p_X < \tilde{R}(T)^{1-\delta}$ for some $\delta>0$. Combined with Theorem~\ref{thm:SROmega} or Theorem~\ref{thm:symmSOmega}, this means that $\omega_u(T)>2$.
\end{proof}

\section*{Acknowledgements} I would like to thank Matthias Christandl, Joshua Grochow, Ryan Williams, Virginia Vassilevska Williams, and Jeroen Zuiddam for helpful discussions and suggestions.

\bibliographystyle{alpha}
\bibliography{papers}

\appendix
\section{Proof that $\omega_u(CW_{q,\sigma}) \geq 2.16805$ for all $q$} \label{appone}

Define the function $f : [0,1/3] \to \R$ by
$$f(v) := \frac{1}{v^v (2/3 - 2v)^{2/3 - 2v} (1/3 + v)^{1/3+v}}.$$
In Section~\ref{secgcw}, we showed that $$\omega_u(CW_{q,\sigma}) \geq \min_{v \in [0,1/3]} 2\frac{\log(q+2) }{ \log(q^{2/3 - 2v} \cdot f(v))}.$$
The value of this optimization problem is computed for $1 \leq q \leq 8$ in a table in Section~\ref{secgcw}, where we see that $\omega_u(CW_{q,\sigma}) \geq 2.16805$ for all $q \leq 8$.

Let $v_q$ denote the argmin for the optimization problem. In particular, for $q=8$, the argmin is $v_8  = 0.017732422\ldots$. From the $q^{2/3 - 2v}$ term in the optimization problem, we see that $v_{q+1} \leq v_q$ for all $q$, and in particular, $v_q \leq v_8$ for all $q > 8$. It follows that $f(v_q) \leq f(v_8) = 2.07389\ldots$ for all $q>8$. Thus, for all $q>8$ we have:
$$\omega_u(CW_{q,\sigma}) \geq \min_{v \in [0,1/3]} 2\frac{\log(q+2) }{ \log(q^{2/3 - 2v} \cdot f(v_8))} = 2\frac{\log(q+2) }{ \log(q^{2/3} \cdot f(v_8))}.$$
This expression equals $2.18562\ldots$ at $q=9$, and is easily seen to be increasing with $q$ for $q>9$, which implies as desired that  $\omega_u(CW_{q,\sigma}) \geq 2.16805$ for all $q \geq 9$ and hence all $q$.

\end{document}